%% file: KLSsubmit.tex
\documentclass[a4paper,11pt,reqno]{amsart}
\usepackage{enumerate}
\usepackage{graphicx,color}
\usepackage{amsfonts}
\usepackage{amsmath}
\usepackage{amssymb}
\newcommand{\Z}{{\mathbb Z}}
\newcommand{\RR}{{\mathbb R}}

\newcommand{\N}{{\mathbb N}}

\providecommand{\C}[1]{\mathcal{#1}}

\DeclareMathOperator{\supp}{supp}

\newtheorem{theorem}{Theorem}[section]

\newtheorem{coro}[theorem]{Corollary}

\newtheorem{definition}[theorem]{Definition}
\newtheorem{example}[theorem]{Example}
\newtheorem{examples}[theorem]{Examples}

\newtheorem{lemma}[theorem]{Lemma}

\newtheorem{prop}[theorem]{Proposition}
\newtheorem{remark}[theorem]{Remark}

\newenvironment{Proof}{{\noindent\emph Proof}\;}
{\hfill$\square$\par\medskip} \newlength\headseptemp

\newcommand{\Hmm}[1]{\leavevmode{\marginpar{\tiny%
$\hbox to 0mm{\hspace*{-0.5mm}$\leftarrow$\hss}%
\vcenter{\vrule depth 0.1mm height 0.1mm width \the\marginparwidth}%
\hbox to 0mm{\hss$\rightarrow$\hspace*{-0.5mm}}$\\\relax\raggedright #1}}}
\begin{document}
\title[One-dimensional continuum models of quasicrystals]{Delone measures of finite local complexity and applications to spectral theory of one-dimensional continuum models of quasicrystals}
\author[]{Steffen Klassert$^1$}
\author[]{Daniel Lenz$^2$}
\author[]{Peter Stollmann$^3$}
\address{$^2$ Mathematisches Institut, Friedrich-Schiller-Universit\"at Jena, D-07743 Jena
Germany; URL: http://www.analysis-lenz.uni-jena.de/
$^3$ Fakult\"at f\"ur
  Mathematik, Technische Universit\"at, 09107 Chemnitz, Germany}
\begin{abstract} We study measures on the real line and present various versions of what it means for such a measure to take only finitely many values. We then study perturbations of the Laplacian by such measures. Using Kotani-Remling theory, we show that the resulting operators have empty absolutely continuous spectrum if the measures are not periodic. When combined with Gordon type arguments this allows us to prove purely singular continuous spectrum for some continuum models of quasicrystals.

\end{abstract}
\date{\today} %
\maketitle

\begin{center}
\emph{Dedicated to Hajo Leschke on the occasion of his 65th birthday}
\end{center}

\section*{Introduction}
Quasicrystals are interesting objects, from a physical and a mathematical point of view. For the physics we refer to the pioneering work of Shechtman, Blech, Gratias and Cahn \cite{ShechtmanBGC-84}, published some 25 years ago, as well as to \cite{BaakeM-00,Janot-92,Senechal-95} and the references therein. The mathematical interest has a number of sources. Here, we are mainly concerned with the fact that quantum mechanical models of quasicrystal exhibit a number of ``strange phenomena'': there is a tendency that the respective Hamiltonians have Cantor sets as spectra, with absence of absolutely continuous spectrum as well as absence of point spectrum. Moreover, anomalous transport is typical for these operators. Very roughly speaking, this can be motivated by the facts that quasicrystals are aperiodically ordered. They are, in a sense, close to periodic structures, which explains the absence of point spectrum. On the other hand, they can also be thought of as random operators, which explains the absence of absolutely continuous spectrum.

Starting with \cite{BIST,Sut,Sut2}, the general paradigm of singular continuous Cantor spectrum has been verified in the discrete one-dimensional setting for various examples, see e.g. the surveys  \cite{Damanik,Len} and references therein. The only results for multidimensional models we are aware of are from  \cite{LenzStollmann1} and the only results in the one-dimensional continuum setting are from \cite{Klassert}. Here, we present a considerable extension of the latter results. In the discrete case, absence of absolutely continuous is a consequence of Kotani theory, \cite{Kotani1,Kotani2}. Absence of eigenvalues can be shown with Gordon-type arguments, see \cite{Damanik} for a survey and a list of relevant references.

It is quite reasonable to try similar strategies for continuum models; one part of that program, the respective Gordon-type arguments have already been established in \cite{DamanikStolz}. So the main point is to prove absence of absolutely continuous spectrum by an appropriate version of Kotani's theory. This had been done in the thesis of one of us, \cite{Klassert} within the framework of Kotani theory for random operators.  In the meantime a deterministic version of Kotani theory has been developed by Remling \cite{Remling2}.
 In the present work we rely upon Remling's oracle theorem,  that allows a purely deterministic approach. The outcome is Theorem \ref{main} below that can also be applied in the random setting  of measure dynamical systems.

We start with an account on Schr\"odinger operators with measure potentials since we adopt the generality of \cite{Remling2} in the following section. There, we will also state the above mentioned oracle theorem, our main tool to prove absence of absolutely continuous spectrum.
 We give a little discussion of finite local complexity for point sets in Section \ref{Sec2}.
 Then, we present  several different versions of a concept of finite local complexity in the context of potentials in Section \ref{Sec2.5} and study their relationship.  This study may be of independent interest.  The main result is presented and proved in Section \ref{Sec3}. Section \ref{Sec4} deals with  applications of four results to measure dynamical systems. A systematic way to generate such systems from subshifts is discussed in Section \ref{Sec5}. The final Section \ref{Sec6} is then devoted to a discussion of singular continuous spectrum for certain models.

\section{Measure perturbations and the oracle theorem.}
\label{Sec1}
We follow \cite{BaakeL-05} and say that a signed Radon measure $\mu$ on $\RR$ is \textit{translation bounded} if there is an open set $U$ with $\sup_{t\in\RR}  |\mu|(U+t)<\infty$; we write $\C{M}$ for the space of all signed Radon measures and $\C{M}^\infty$ for the subspace of all translation bounded measures. Here $|\mu|$ denotes the variation of $\mu$. It is easy to see that
$$
\C{M}^\infty=\bigcup_{C\ge 0}\C{V}^C ,
$$
where the latter is defined as in \cite{Remling2}, i.e.,
$$
\C{V}^C =\{\mu\in \C{M}\mid |\mu|(J)\le C\max\{|J|,1\}\mbox{  for all compact intervals }J\subset \RR\} .
$$
As explained in \cite{Remling2}, there is a metric $d$ on $\C{V}^C$ that will be used later on. The shift is denoted differently here: we let $T_x\mu(A):= \mu (A-x)$ so that $\supp T_x\mu= \supp\mu +x$; in particular, $T_x=S_{-x}$ for the shift $S_x$ used in \cite{Remling2}. In the latter article, the author relies upon \cite{AmorR-05}, as far as the definition of $-\Delta+\mu$ is concerned. Here we sketch a more traditional way via quadratic forms that is pretty easy and direct, since translation bounded measures are form small with respect to $-\Delta$ in one dimension. Here, of course, the form $h_0$ associated to $-\Delta$ is defined on $W^{1,2}(\RR)$  by
$$h_0 (u) := \int_\RR u'(t)\bar{u}'(t)dt.$$

\begin{lemma}
 Let $\mu\in\C{M}^\infty$.
  Then $\mu$ is form small with respect to $-\Delta$. In particular,
$$D(h_0+\mu):= W^{1,2}(\RR), (h_0+\mu)[u,v]:= \int_\RR u'(t)\bar{v}'(t)dt +\int_\RR u(t)\bar{v}(t)d\mu(t)
$$
defines a closed semibounded form. The associated selfadjoint operator is denoted by $-\Delta+\mu$.
\end{lemma}
Let us sketch the proof:

We use the Sobolev embedding theorem that implies that every $u\in W^{1,2}(\RR)$ has a continuous representative and use this repesentative in
$$\mu[u,v]:=\int_\RR u(t)\bar{v}(t)d\mu(t) .
$$
That the so defined form
is form small follows from the Sobolev embedding theorem as well. See \cite{Kuchment-04}, Lemma 8 for an appropriate version that gives, far any $\gamma>0$ a $C_\gamma$ such that
\begin{eqnarray*}
 |\mu[u,u]| &=&  \int_\RR |u(t)|^2 d\mu(t)\\
&=&\sum_{k\in\Z} \int_{[k,k+1]} |u(t)|^2 d\mu(t)\\
&\le&\sum_{k\in\Z} C\| u|_{[k,k+1]} \|_\infty^2\\
&\le&\sum_{k\in\Z} C(\gamma \| u'|_{[k,k+1]} \|_2^2+C_\gamma \| u |_{[k,k+1]} \|_2^2)\\
& =& C(\gamma h_0[u,u] +C_\gamma \| u  \|_2^2) .
\end{eqnarray*}
That form small perturbations of closed forms yield closed forms is a fundamental truth, see the KLMN-Theorem in \cite{ReedS-75} or \cite{Faris-75}, Part I, Paragraph 5.

 \begin{remark}

\begin{itemize}
\item[(1)] In \cite{AmorR-05} an alternative direct definition of Schr\"odinger operators with measure perturbations is given. It has the advantage that the negative part of the measure need not be relatively bounded with respect to the Laplacian, so that operators can be treated that are unbounded from below.
\item[(2)] The form method can be used in any dimension; see \cite{Stollmann-92,StollmannV-96} for details and references.
\item[(3)] That these two different methods lead to the same operators is not explicitly stated in \cite{AmorR-05}. It follows, however, from the reasoning in the proof of Theorem 6.1 in \cite{Remling2}.
\end{itemize}
\end{remark}

We now state Remlings Oracle Theorem from \cite{Remling2} for the readers convenience:
\begin{theorem}[Oracle Theorem]
Let $A\subset \RR$ be a Borel set of positive measure, and let $\varepsilon>0$, $a,b\in\RR$, $a<b$, $C>0$. Then there exist $L>0$ and a continuous function (the oracle)
$$
\triangle:\C{V}^C_{(-L,0)}\to \C{V}^C_{(a,b)}
 $$
so that the following holds. If $\mu\in\C{V}^C$ and the half line operator $H=-\Delta +\mu\upharpoonright_{[0,\infty)}$ satisfies $\sigma_{ac}(H) \supset A$, then there exists an $x_0>0$ so that for all $x\ge x_0$, we have
$$
d\left( \triangle\left( 1_{(-L,0)}S_x\mu\right) , 1_{(a,b)}S_x\mu\right) <\varepsilon .
$$
\end{theorem}

\section{Notions of finite local complexity for sets.}\label{Sec2}
We describe the positions of ions in terms of a discrete set $D\subset\RR^d$. In case of a perfect crystal, $D$ would be \emph{periodic}, i.e., invariant under translations by elements of a lattice. Here lattice means a discrete subgroup that spans $\RR^d$ as an $\RR$-vector space. Recall that a \textit{Delone set} $D\subset \RR^d$ is characterized by two properties that it shares with periodic sets: It is uniformly discrete and relatively dense i.e.

\begin{itemize}

\item  there is $r>0$ such that $U_r(x)\cap U_r(y)=\emptyset$ for $x,y\in D, x\not=  y$,

\item   and there exists $R>0$ such that $\bigcup_{x\in D}B_R(x)=\RR^d$.

\end{itemize}
Here, the open and closed balls of radius $s>0$ around $x\in \RR^d$ are denoted by $U_s (x)$ and $B_s (x)$ respectively. Delone sets may be used to describe the positions of ions in solids that are not as ordered as crystals. The following property is clearly true for periodic sets:

A discrete set $A\subset\RR^d$ is said to be of \emph{finite local complexity}, provided the following holds: for any $L\ge 0$
$$
\{ B_L(x)\cap A-x\mid x\in A\}
$$
is a finite set of subsets of $\RR^d$. We are mainly concerned with subsets of the real line in which case the property means that in any interval $[x-L,x+L]$, $x\in A$ only finitely many different pattern occur. We do not go into formalizing the concept of pattern here, see \cite{LenzS-01} for a thorough discussion and Figure \ref{fig:1} for an illustration.
\begin{figure}
 \centering
 \includegraphics[width=11cm,height=2cm]{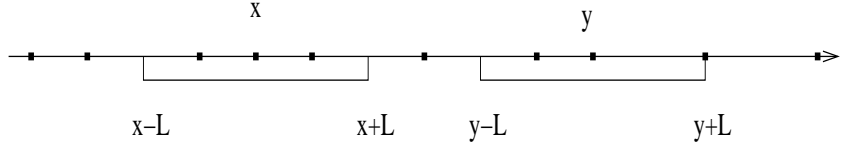}
 \caption{A Delone set of finite local complexity}
 \label{fig:1}
\end{figure}

Since the geometry of $\RR$ is particularly simple, we can easily verify the following observation.
\begin{remark}
\begin{rm} \begin{itemize}
  \item [(1)] A closed set $D\subset\RR$ is a Delone set if it can be written as $D=\{ x_n\mid n\in \Z\}$ where $x_{n+1}-x_n\in [2r, \frac{R}{2}]$; it is of finite local complexity if, additionally,  $x_{n+1}-x_n$ can only take finitely many values. (For investigation of higher dimensional analogues to these facts, we refer to \cite{Lag}.)
\item [(2)] A set of finite local complexity is uniformly discrete but not necessarily relatively dense.
 \end{itemize}              \end{rm}
\end{remark}

Therefore, it is different whether one bases a concept of finite local complexity on ``what can be seen around a given point''  or on ``what comes next''. This is further illustrated by the following examples.

\begin{examples}\begin{rm}
 \begin{itemize}
  \item[(1)]  Any periodic set is a Delone set of finite local complexity.
  \item[(2)] Any subset of a set of finite local complexity has finite local complexity, e.g., any subset of $\Z$.
 \end{itemize}\end{rm}
\end{examples}

It will be convenient to deal with colored Delone sets as well. Let $A$ be a finite set. A colored Delone set with colors from $A$ is a pair $(D,a)$ consisting of a Delone set $D$ and a map
$$ a : D\longrightarrow A.$$
The restriction of  a colored Delone set $D' =(D,a)$  to a subset $B$ is defined to be the pair
$$ D'\cap B:=(D\cap B,  a: D\cap B\longrightarrow A).$$
A colored Delone set is said to have finite local complexity if the underlying Delone set has finite local complexity. As there are only finitely many elements in $A$ this can easily be seen to be equivalent to finiteness of the set of restriction
$$ \{ (D' - x) \cap B_R : x\in D\}.$$

\section{Notions of finite local complexity for measures} \label{Sec2.5}
The kind of potentials we are interested in are of the form
$$
V_D(t)=\sum_{x\in D}V(t-x) ,
$$
where $V\in L^1$ is a given single site potential of compact support and $D$ is uniformly discrete. More generally, one may be interested in potentials of the form
$$V_{(D,a)}:=\sum_{x\in D} V_{a(x)} (t -x)$$
where $(D,a)$ is a Delone colored by the finite $A$ and  $V_a, a\in A$, are $L^1$ potentials with compact support.  From the analytical point of view, it does not make much difference to include finite signed measures of compact support as single site potentials; a typical piece
of such a potential is depicted in Figure \ref{fig:2}.

\begin{figure}[h!]
 \centering
 \includegraphics[width=12.5cm,height=1cm]{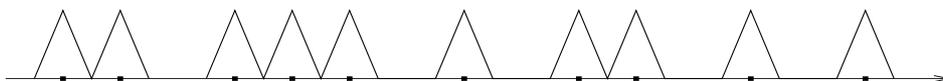}
 \caption{A typical potential $V_D$}
 \label{fig:2}
\end{figure}

If we do not want to use the construction of $V_D$, how can we formalize the finite local complexity it inherits from $D$? The problem we encounter here is the fact that although from a certain perspective, $V_D$ is made of a finite number of parts, we have to pick the right parts. If we slice the real line into intervals in different ways, we end up with an infinite supply of different forms. For point sets as above, this difficulty does not occur, as the point set itself gives a natural grid. Note that point sets can naturally be interpreted as measures: any uniformly discrete $D$ induces the translation bounded measure
$$
\delta_D:= \sum_{x\in D}\delta_x .
$$
Note that $V_D$ defined above is just given as $\delta_D \ast V$, where $\ast$ denotes the usual convolution. Similarly, we may define a convolution between  Delone sets colored by $A$ and tuples of functions $V_a$, $a\in A$, by
$$ {(D,a)} \ast V := V_{(D,a)}.$$

We now present a number of different concepts one might come up with. An ultimate criterion of its usefulness for our present purpose lies in the applicability to spectral theory we have in mind, especially whether it can be used together with the oracle theorem, mentioned above. In our context, it will turn out that the simple
finite decomposition property s.f.d.p. given in Definition \ref{sfdp} gives what we need for our purposes.

\medskip

We start with some general notation used in our dealing with measures and their local restrictions:
 A \emph{piece} is a pair $(\nu,I)$ consisting of a closed interval $I$ with length $|I|>0$ and a measure $\nu$ supported on $I$; we write $\nu^I$ as shorthand notation, or sometimes only $\nu$, which is, of course, an abuse of notation. If $\mu$ is a measure on $\RR$ we say that $\nu^I$ \emph{occurs in} $\mu$ \emph{at} $x$ if $(1_{[x,x+|I|]}\mu,[x,x+|I|]) $ is a translate of $(\nu, I)$ (in the obvious sense).
We call $|I|$ the length of the piece $\nu^I$. A \emph{finite piece} is a piece with finite length. The \emph{concatenation} $
\nu^I=\nu_1^{I_1}|\nu_2^{I_2}|...$ of a finite or countable family $(\nu_j^{I_j})_{j\in N}$, $N=\{ 1,2,...,|N|\}$ of  finite pieces is easy to visualize and somewhat complicated to write down:
$$
I=[\min I_1, \min I_1+\sum_{j\in N} |I_j| ]
$$
$$
\nu= \nu_1+ \sum_{k\in N, k\ge 2}T_{(\min I_1+\sum_{l=1}^{k-1}|I_l|-\min I_k)}\nu_k .
$$
In this case we also say that $\nu^I$ \emph{is decomposed into}  $(\nu_j^{I_j})_{j\in N}$. Then, $\nu_j^{I_j}$ is the measure that occurs at $x$ for
$$x \in [ \min I_j +\sum_{k< j} |I_j|, \max I_j +\sum_{j< j} |I_j|].$$
See Figure \ref{fig:4} for an illustration.
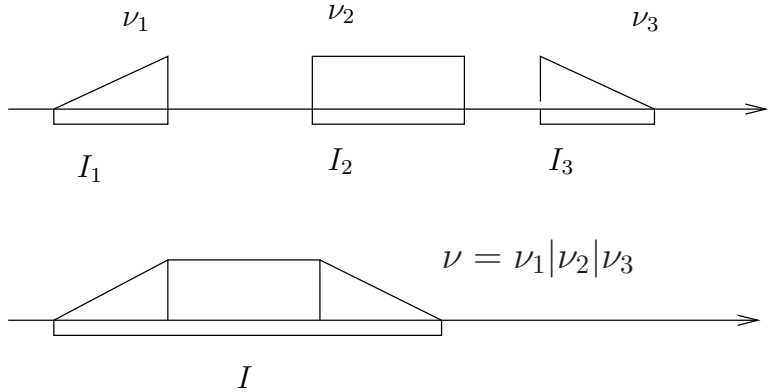
\begin{figure}[h!]
 \centering
 \input{FIG4.pspdftex}
 \caption{The concatenation of pieces.}
 \label{fig:4}
\end{figure}
What the formulae express is easy to describe: we translate the piece $\nu_2^{I_2}$ in such a way that it starts where $
 \nu_1^{I_1}$ ends. Then we pick the next piece and put it right to the right of $\nu_2^{I_2}$, and so on.

\smallskip

Now, here comes one version of finite local complexity for measures.

\begin{definition}
 A signed measure $\mu$ has a \emph{finite number of local pieces} with parameter $\rho>0$, abbreviated \emph{f.l.p} if  for every $x\in\RR$ there is $x^\prime\in \RR$ with $| x - x'|\leq \rho$  such that,  for every $L\ge  2 \rho$,
$$
\{ 1_{[-L,L]}(T_{x^\prime}\mu)\mid x\in\RR\}
$$
is a finite set of measures. We say that $\mu$ is \emph{finally f.l.p}, if there is $x_0\in\RR$ such that $1_{[x_0,\infty)}\mu$ is f.l.p.
\end{definition}

\begin{remark}\label{rem2.1}
 \begin{itemize}\begin{rm}
  \item [(1)] If $D$ is a Delone set of finite local complexity, then $\delta_D$ has f.l.p. In fact, we can pick $\rho=R$ and $x^\prime\in D\cap [x-\rho,x+\rho]$, the latter set being nonempty due to the relative denseness condition.
   \item [(2)] If $D$ is any subset of a periodic set, then $\delta_D$ has f.l.p. In particular, there are f.l.p point measures such that the underlying sets are no Delone sets. \end{rm}

 \end{itemize}
\end{remark}

Another version of finite local complexity is given as follows.

\begin{definition} A  measure $\mu$ has the \emph{finite decomposition property}, abbreviated f.d.p, if there exists
a finite set $\C{P}$ of finite pieces, called local pieces,  and $x_0\in\RR$ such that $1_{[x_0,\infty)}\mu^{[x_0,\infty)}$ is the translate of a concatenation $\nu_1^{I_1}|\nu_2^{I_2}|...$, where the $\nu_j^{I_j}\in \C{P}$, for all $j\in\N$.
We denote $x_1:=x_0+|I_1|, x_2:= x_1+ |I_1|$, ... and call the set $\{x_0,x_1,...\}$ the \emph{grid} of the decomposition.
Without restriction we can and will assume that $\min I=0$ for every $\nu^I\in\C{P}$. We denote by $l_ \C{P}$, $L_\C{P}$ the minimal and maximal length of pieces in $\C{P}$, respectively.
\end{definition}

For our main result we will need a stronger version of f.d.p.

\begin{definition}\label{sfdp} A $\mu$ has the \emph{simple finite decomposition property}, abbreviated s.f.d.p,  if it has the f.d.p with decomposition such that  there is $\ell >0$ with the following property: assume that two pieces $\nu_{-m}^{I_{-m}}|...|\nu_0^{I_0}|\nu_1^{I_1}|\nu_2^{I_2}|...|\nu_{m_1}^{I_{m_1}}$ and
$\nu_{-m}^{I_{-m}}|...|\nu_0^{I_0}|\mu_1^{J_1}|\mu_2^{J_2}|...|\mu_{m_2}^{J_{m_2}}$ occur in the decomposition of $\mu$ with a common first part
$\nu_{-m}^{I_{-m}}|...|\nu_0^{I_0}$ of length at least $\ell$ and such that
$$
1_{[0,\ell)}(\nu_1^{I_1}|\nu_2^{I_2}|...|\nu_{m_1}^{I_{m_1}})=1_{[0,\ell)}(\mu_1^{J_1}|\mu_2^{J_2}|...|\mu_{m_2}^{J_{m_2}}),
$$
where the $\nu_j^{I_j}$, $\mu_k^{J_k}$ are from the decomposition (in particular, all the pieces belong to $\C{P}$ and so they start at $0$) and the latter two concatenations are of lengths at least $\ell$. Then
$$
\nu_1^{I_1}=\mu_1^{J_1} .
$$
\end{definition}

Note, that instead of the equality of the pieces we ask for in the last equation of the definition it suffices to ask for the equality of the length of the pieces.

\bigskip

Before we illustrate the preceding definition with examples let us explain the difficulties that might occur and prevent a measure with f.d.p from satisfying s.f.d.p. Essentially there are two problems: one is periodicity and the other is local constancy. Let us discuss periodicity first:  Among our local  pieces there may be  $\nu$ and $\nu'$ with  $\nu'$ equal to a  concatenation of finitely many copies of $\nu$. Then, a high power of $\nu$ can not be decomposed uniquely. This problem may be circumvented by deleting $\nu'$ from our set of local pieces. We now turn to local constancy. You might decompose a fixed multiple of Lebesgue  measure - e.g. the measure  zero -  in pieces of two different lengths, say, in such a way that it does not have the s.f.d.p (whereas f.d.p holds): take one short piece and then a long one, then two short pieces and then a long one, etc. For a decomposition of that fashion, you can never predict the appearance of a long piece. Note that this situation is not covered by the periodic situation (as the lengths of the pieces may be rationally independent). In some sense these are the only ways how a measure with f.d.p can fail to satisfy s.f.d.p:

\begin{prop}\label{hilf} Let $\mu$ be a measure with f.d.p. with respect to the local pieces  $\nu_1^{I_1},\ldots, \nu_N^{I_N}$. If $\mu$ does not have s.f.d.p. with respect to any  set of local pieces, then there must be two local pieces among the $\nu_j^{I_j}$, which are multiples of Lebesgue measure i.e. have the form $ c\;  1_{I} \;  dx$ for suitable constants $c$ and intervals $I$.
\end{prop}
\begin{Proof} Denote by $D_j$ the set of occurrences of $\nu_j^{I_j}$ in $\mu$. (Here, we count all occurrences irrespective whether these occurrences match with the given decomposition.) Then,
$D:= \cup D_j$ is relatively dense, as there  exists a decomposition of $\mu$.  We consider two cases.

\smallskip

\textit{ $D$ is uniformly discrete.} In this case,  we can just recode $\mu$ in the following way: Let
$$ \ldots x_{-1} < x_0 < x_1 < \ldots$$
be the points of $D$.
Then, we can decompose $\mu$ according to
$$ \ldots \mu|_{[x_{-1},x_0)} \mu|_{[x_{0},x_1)}\mu|_{[x_{1},x_2)}\ldots$$
As $\mu$ has f.d.p  with respect to the $\nu_j^{I_j}$, this is indeed a decomposition into finitely many different pieces. By construction the occurrence of each piece is determined locally i.e. only depends on the restriction of $\mu$ to a certain neighborhood of the piece. This gives the desired simplicity.

\smallskip

\textit{ $D$ is not uniformly discrete.}  Then,  there must exist an   accumulation points among the occurrences of the local pieces. If the occurrences of a fixed  local piece accumulate, then this piece must be locally translation invariant and hence a multiple of the Lebesgue measure. If there are two  local pieces with  accumulation points we therefore obtain two pieces which are multiples of Lebesgue measure.

Thus, it remains to show that if there is only one local piece with accumulation points, then the system still has s.f.d.p.:  In this case, there is only one exceptional  piece which is a multiple of Lebesgue measure. If this is the only piece in the decomposition, s.f.d.p is obviously valid.
Otherwise, we can use a coding by the  occurrences of the other local pieces and break gaps resulting from  multiple occurrences of the exceptional piece  by starting at the leftmost occurrence of the piece in the gap and then proceeding by adding this piece step by step as long as possible. By f.d.p. this will result in only finitely many pieces.
\end{Proof}

\begin{remark} {\rm Let us note that even in non-periodic situations it can be hard to decompose into a fixes set of pieces. This is well known even for Delone sets. For example, for certain Delone sets  coming from so called primitive substitutions uniqueness of some natural decompositions has been investigated (and proven)  in \cite{Sol}  (see \cite{Mos} as well for the corresponding result in symbolic dynamics).}
\end{remark}

\begin{example}
 \begin{rm}
  \begin{itemize}
   \item [(1)] If $\mu$ is finally periodic, i.e., there exist $x_0\in\RR$ and $p>0$ such that $\mu(A+p)=\mu(A)$ for all measurable $A\subset [x_0,\infty)$, then $\mu$ has the s.f.d.p. In fact, $\C{P}:=\{ (1_{[x_0,x_0+p)}\mu, [x_0,x_0+p])\}$ does the job, as one easily verifies.
   \item [(2)] If $D$ is a Delone set of finite local complexity, then $\delta_D$ has the s.f.d.p. In fact, this follows from the  Lemma  belwo plus the fact that such measures have the f.e.p, but we can also define a suitable set $\C{P}$ directly: as remarked in \ref{rem2.1}, there are only finitely many different distances between neighboring points of $D$; denote them by $0<l_1<...<l_n$. We put $I_j=[0,l_j]$, $\nu_j=\delta_o$ and $\C{P}=\{ \nu_j^{I_j}\mid j=1,...,N\}$. Now, for every $x_0\in D$ we find a decomposition; due to the construction of $\C{P}$, every decomposition is simple.

\item[(3)] The measure $\mu$ can be shown to have f.d.p. if and only if there exists a colored Delone $D'$ set and an $A$ tupel $(\nu_a)$, $a\in A$, of measures with compact support such that $\mu =D' \ast \nu$.
  \end{itemize}
 \end{rm}
\end{example}

While the preceding notions were  based on ``what can be seen around a given point'', the next notion formalizes the ``what comes next'' point of view. This notion is stronger than the previous ones:

\begin{definition}
A signed measure $\mu$ is said to have the \emph{finite extension property}, short \emph{f.e.p} if there exist $\rho >0$, $L_0$ and, for any $L\ge L_0$ a finite set $\C{E}^L$ of pieces with interval $[0,L]$ with the following property: for $r\ge\rho$ and any piece $\nu^{[0,r]}$ there exists $r^\prime\in [r-\rho, r]$ such that whenever $\nu^{[0,r]}$ occurs in $\mu$ at $x$, there is a $\hat{\nu}^{[0,L]}\in \C{E}^L$  such that
$(1_{[0,r^\prime]}\nu+T_{r^\prime}\hat{\nu}, [0,r^\prime+L])$ occurs in $\mu$ at $x$.
We say that  $\mu$ is \emph{finally f.e.p} if there is an $x_0$ such that the above property is valid for $x\ge x_0$.
\end{definition}
See Figure \ref{fig:3} for an illustration; here are some remarks:

\begin{figure}[h!]
 \centering
\begin{picture}(0,0)%
\includegraphics{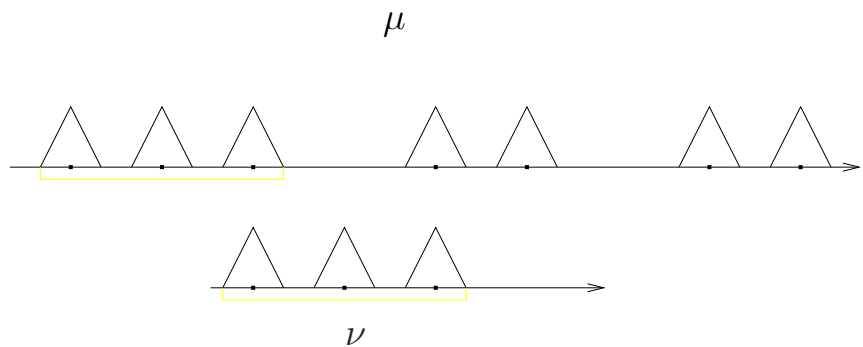}%
\end{picture}%
\setlength{\unitlength}{4144sp}%
\begingroup\makeatletter\ifx\SetFigFont\undefined%
\gdef\SetFigFont#1#2#3#4#5{%
  \reset@font\fontsize{#1}{#2pt}%
  \fontfamily{#3}\fontseries{#4}\fontshape{#5}%
  \selectfont}%
\fi\endgroup%
\begin{picture}(6324,2134)(439,-1934)
\put(2656, 29){\makebox(0,0)[lb]{\smash{{\SetFigFont{14}{16.8}{\rmdefault}{\mddefault}{\updefault}{\color[rgb]{0,0,0}$\mu$}%
}}}}
\put(2431,-1861){\makebox(0,0)[lb]{\smash{{\SetFigFont{14}{16.8}{\rmdefault}{\mddefault}{\updefault}$\nu$}}}}
\end{picture}%

 \caption{A piece $\nu^I$ and its occurrence in $\mu$.}
\label{fig:3}
\end{figure}

\begin{remark}\begin{rm}
 \begin{itemize}
 \item [(1)] Note that there is a difference in how we can define finally f.l.p and finally f.e.p: the measure $1_{[x_0,\infty)}\mu$ is almost never f.e.p. In fact, only if it is zero. Otherwise, there are arbitrarily long intervals where $\mu$ is zero, followed by something nontrivial. On the other hand, if $\mu$ is f.l.p, then  $1_{[x_0,\infty)}\mu$ is f.l.p as well
\item [(2)] If $D$ is a Delone set of finite local complexity, then $\delta_D$ has f.e.p. In fact, we can pick $\rho=R$ and $x^\prime\in D\cap [x-\rho,x+\rho]$, the latter set being nonempty due to the relative denseness condition.
   \item [(3)] If $D$ is any subset of a periodic set, then $\delta_D$ is f.l.p but not necessarily f.e.p.
\end{itemize}
\end{rm}
\end{remark}

Finally, we define a strong form of decomposition property.

\begin{definition}
  A  measure $\mu$ has the \emph{unique decomposition property}, abbreviated u.d.p, if
  it has the finite decomposition property and there exists an $R>0$ such that the  local piece occurring  at an $x$ in the decomposition is uniquely determined by the restriction of $\mu$ to $[x-R, x+R]$.
 \end{definition}

The following lemma gathers simple relations between the various notions of local finiteness introduced so far. It is easy to prove.
\begin{lemma}
$ u.d.p \Longrightarrow f.e.p. \Longrightarrow s.f.d.p \Longrightarrow f.d.p \Longrightarrow f.l.p.$
\end{lemma}

The converse implications do not hold in general: For example Lebesgue measure may easily be decomposed periodically yielding s.f.d.p, still it does not satisfy $u.d.p.$.  Also, we have already discussed how measures with  f.d.p may  not have s.f.d.p. Finally, considering measures coming from point sets with increasing gaps with incommensurate lengths, one can construct examples of measures with  f.l.p. but without f.d.p. However,  assuming  relative denseness of local pieces we can use  'Deloneification' to show  that $f.l.p$ implies $u.d.p.$. This is the content of the next theorem. The proof is similar to the proof of Proposition \ref{hilf} above.

\begin{theorem} Let $\mu$ be a measure on $\RR$ with $f.l.p.$ with parameter $\rho>0$. If for some $L\geq 2 \rho$ the set of occurrences of any local piece of length $L$ is relatively dense, then $\mu$ is either a multiple of Lebesgue measure or has the unique decomposition property.
\end{theorem}
\begin{Proof} We consider two cases.

\smallskip

\textit{Case 1: There is a local measure $\nu$ whose occurrences  have a minimal distance $r>0$.}  In this case the occurrences of $\nu$ in $\mu$ form a Delone set. We can use this Delone set to  do a recoding of $\mu$ by occurrences of $\nu$. More precisely, label the occurrences of $\nu$ in $\mu$ by
$$ \ldots, x_{-1} < x_0 < x_1 < x_2 < x_3 \ldots$$
Then, we can decompose $\mu$ according to
$$ \ldots \mu|_{[x_{-1},x_0)} \mu|_{[x_{0},x_1)}\mu|_{[x_{1},x_2)}\ldots$$
As $\mu$ has f.l.p and there is a minimal distance between occurrences of $\nu$ the Delone set has finite local complexity and  this is indeed a decomposition into finitely many different pieces. By construction the occurrence of each piece is determined locally i.e. the unique decomposition property holds.

\smallskip

\textit{Case 2: The set of occurrences of each local measure $\nu$ has arbitrary small gaps.} The assumption means that  for each local measure $\nu$
$$\nu = S_x \nu |_{\supp \nu}$$
for a set of nonzero $x$ with $x\to 0$. Then, $\nu$ is (locally) translation invariant and hence is a multiple of Lebesgue measure. Different local measures must be the same multiple as $L \geq 2 \rho$.
\end{Proof}

\begin{coro} Let $\mu$ be a measure on $\RR$ with $f.l.p.$ with parameter $\rho$ such that for some $L\geq 2 \rho + S$ the set of occurrences of any local piece of length $L$ is relatively dense. Let $\nu$ be a measure of compact support contained in $[0,S]$. Then, $\mu \ast \nu$ is either Lebesgue measure or has unique decomposition property.
\end{coro}
\begin{Proof} As $\mu$ has $f.l.p$ and $\nu$ has compact support,  the measure  $\mu \ast \nu$ has f.l.p. as well. By assumption  the measure $\mu\ast \nu$ has relative dense set of occurrences  for local measures of length $L$. Now, the statement follows from the previous theorem.
\end{Proof}

If the underlying measure $\mu$ comes from a Delone set of finite local complexity a somewhat stronger statement holds.

\begin{prop} \label{delone-sfdp}
Assume that $D$ is a Delone set of finite local complexity and $\nu$ is a finite signed measure of compact support. Then
$$ \mu=\sum_{x\in D}T_x\nu=\nu * \delta_D $$
has the s.f.d.p.
\end{prop}
Unfortunately, the proof is not as evident as the statement may seem. Here is the problem: You cannot recover $D$ from $\nu*\delta_D$ in general. Easy periodic examples will convince you of this unpleasant fact. This lead to additional restrictions  in \cite{Klassert}  concerning the single site bumps that were allowed. Here, we can circumvent these additional assumptions, because we can still choose the particular decomposition according to our purposes.

\begin{proof}
 By shifting $D$ and $\nu$ we may suppose that $\min(\supp \nu)=0$; moreover, $\max(\supp \nu)=:s >0$ without restriction. Let $D=\{ x_k|\mid k\in\Z\}$ with the $x_k$ strictly increasing. We decompose $\mu$ starting from $x_0$ and take $ \{ x_k|\mid k\in\N_0 \}$ as a grid. Note that
$$
1_{[x_k,x_{k+1}]}\mu = \sum_{x\in D\cap[x_k-s,x_k]}T_x\nu =T_{x_k}\sum_{x\in D\cap[x_k-s,x_k]}T_{x-x_k}\nu,
$$
with a finite number of different possibilities for $x-x_k$ in $[x_k-s,x_k]$, independently of $x_k$. Therefore,
$$
\C{P}:=\{ (\sum_{x\in D\cap[x_k-s,x_k]}T_{x-x_k}\nu, [0, x_{k+1}-x_k])\mid k\in\Z, k\ge 0 \}
$$
is a finite set of finite pieces that decomposes $1_{[x_0,\infty)}\mu$.
Now we want to show that the decomposition we chose has the s.f.d.p. Choose $l>\max\{ s, L_{\C{P}}\}$ and assume that two pieces
 $\nu_{-m}^{I_{-m}}|...|\nu_0^{I_0}|\nu_1^{I_1}|\nu_2^{I_2}|...|\nu_{m_1}^{I_{m_1}}$ and
$\nu_{-m}^{I_{-m}}|...|\nu_0^{I_0}|\mu_1^{J_1}|\mu_2^{J_2}|...|\mu_{m_2}^{J_{m_2}}$ occur in the decomposition of $\mu$.

Denote the point in the grid where $\nu_1^{I_1}$ starts by $x$ and the point where $\mu_1^{J_1}$ starts by $x^\prime$. Since $\mu$ is given as a sum of translates depending on what is happening to the left, we get that
$$
1_{[x,x+\ell_{\C{P}})}\mu = T_{x^\prime -x}1_{[x^\prime,x^\prime+\ell_{\C{P}})}\mu .
$$
Subtracting the common contribution from $D\cap [x-s,x]$ and $D\cap [x^\prime-s,x^\prime]$ we get
$$
T_{| I_1|}\nu = T_{| J_1|}\nu,
$$
and the asserted equality of the length $| I_1|=|J_1|$ follows.
\end{proof}

The considerations of this section suggest the following as definition for a Delone measure of finite local complexity.

\begin{definition} A measure on $\RR$ is said to be a Delone measure of finite local complexity if there exist finitely many measures $\nu_1,\ldots, \nu_N$ with compact support such that with the sets $D_j$ of occurrences of $\nu_j$ in $\mu$, $j=1,\ldots, N$, the following holds:

\begin{itemize}
\item  The set $\cup_j D_j$ is a Delone set of finite local complexity.
\item For any $x\in\supp \mu$, there exists a $j\in\{1,\ldots, N\}$ and an $p\in D_j$ such that $x$ belongs to $p+  \supp \nu_j$ and $\mu$ agrees with $T_p \nu_j$ on $p + \supp \nu_j$.
\end{itemize}
\end{definition}

Let us note that this definition can be used verbatim to deal with measures in arbitrary dimensions.

\section{The main result}\label{Sec3}
In this section, we present the main result on absence of absolutely continuous spectrum for aperiodic measure of 'finite local complexity'.

\begin{theorem}
 \label{main}
Let $\mu\in\C{M}^\infty$ be a measure that has the s.f.d.p and assume that the absolutely continuous spectrum of the half-line operator $-\Delta +\mu \upharpoonright_{[0,\infty)}$ is not empty. Then,
$\mu$ is eventually periodic.
\end{theorem}
\begin{proof}
 The idea goes as follows: due to the s.f.d.p there are pieces of arbitrary length that appear infinitely often. If we pick such a piece of sufficient length, the oracle theorem allows us to predict the next piece up to a small error. Since the next piece itself comes from a finite set, it is uniquely determined, provided the error is small enough. That means that after the two occurences of the given long piece, the same shorter piece follows. Iterating the procedure we get longer and longer translates of the measure that agree, which means that the measure is eventually periodic.

Here are the details: assume that $A = \sigma_{ac}(-\Delta +\mu)\not=\emptyset$ is nonempty. Let $G:=\{ x_k\mid k\in\N_0\}$ be the grid and $\C{P}$ the set of pieces of a decomposition with the s.f.d.p and $\ell$ the according length.

Denote
$$
\C{D}_\ell:=\{ 1_{(-1,\ell)}[T_{-x}\mu]\mid x\in G\}
$$
and note that this set is finite, due to the s.f.d.p. Consequently, there is $\varepsilon >0$ with $d( \nu_1,\nu_2) >2\varepsilon$ for different $\nu_1,\nu_2\in \C{D}_\ell$.

This fixes $\varepsilon$ for our application of the oracle theorem; moreover, we choose $a=-1, b=\ell$. Pick $L$ according to the oracle theorem, without restriction $L>\ell$.

We construct a coarser grid $G_L\subset G$ by $y_0:=x_0$ and $y_{k+1}-y_k\in [L,L+L_{\C{P}}]$. Since $G\cap [y_k,y_{k+1}]$ is finite, it is clear that
$$
\C{P}_L:=\{ \left( T_{-y_k}[1_{[y_k,y_{k+1}]}\mu], T_{-y_k}([y_k,y_{k+1}]\cap G)\right)\mid k\in\N\}
$$
is finite. Consequently, there is at least one $k$, for which infinitely many translates of $1_{[y_k,y_{k+1}]}\mu$ appear in $\mu$ so that the corresponding parts of $G$ are translates of $G\cap [y_k,y_{k+1}]$. Call $z^1:= y_{m+1}$, where $y_m$ is  one of the
corresponding points in $G_L$, $m>k$ and denote $y^1:=y_{k+1}$; see the following Figure \ref{fig:5}.
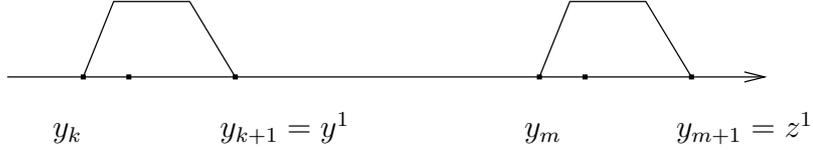
\begin{figure}[h!]
 \centering
 \input{FIG5.pspdftex}
 \caption{Illustration of the proof.}
 \label{fig:5}
\end{figure}
We apply the oracle theorem with $x=y^1$ and $x=z^1$ and get that
$$
d(1_{(-1,\ell)}T_{y^1}\mu, 1_{(-1,\ell)}T_{z^1}\mu)\le 2\varepsilon .
$$
By our choice of $\varepsilon$ this implies
$$
1_{(-1,\ell)}T_{y^1}\mu=1_{(-1,\ell)}T_{z^1}\mu .
$$
By definition of $\C{P}_L$ we moreover know that the pieces of $\mu$ starting at $y_k$ and $y_m$, respectively, are decomposed in the same way. The s.f.d.p (note that $L>\ell$) guarantees that the pieces start
 at $y^1$ and $z^1$ are translates of each other, so that the next point $y^2:= min(G\cap (y^1,\infty))$ in $G$ following $y^1$ and $z^2:= min(G\cap (z^1,\infty))$ satisfy
$$
z^2-z^1 =y^2-y^1 \mbox{  and  } T_{-y^1}[1_{[y^1,y^2]}\mu]= T_{-z^1}[1_{[z^1,z^2]}\mu] .
$$
This implies that $1_{[y_k,y^2]}\mu$ is a translate of $1_{[y_m,z^2]}\mu$. We can now iterate the procedure to find $y^3,z^3\in G$ such that $1_{[y_k,y^3]}\mu$ is a translate of $1_{[y_m,z^3]}\mu$. Since, moreover, $y^{j+1}-y^j,z^{j+1}-z^j\ge \ell_{\C{P}}$, we get that $1_{[y_k,\infty)}\mu$ is a translate of $1_{[y_m,\infty)}\mu$. In other words, $\mu$ is finally periodic.
\end{proof}

\section{Application to measure dynamical systems.}\label{Sec4}
In this section we consider applications of the main result to dynamical systems of measures.

\bigskip

A  dynamical system $(\varOmega,S)$ consists of a compact set   $\varOmega$  and a continuous action $S$ of $\RR$ on $\varOmega$.
A  dynamical system  is called minimal if
$$\{S_x \mu : x\in \RR\}$$
is dense in $\varOmega$ for every $\mu \in \varOmega$. An $S$ invariant probability measure on a  dynamical system $(\varOmega,S)$  is called ergodic if any invariant measurable subset of $\varOmega$ has measure zero or measure one. If there exists a unique invariant probability measure, then $(\varOmega, T)$ is called uniquely ergodic.  An element $\omega\in \varOmega$ is called periodic with period $x\neq 0$ if $S_x \omega = \omega$. The, dynamical system is called
 periodic if all of its elements are periodic with the same period.

We will be concerned with measure dynamical systems. In this case $\varOmega$ is a closed (and hence compact)  subset of
 $\mathcal{V}^C$ for some $C>0$, which is invariant under the shifts $S=T$.
A measure dynamical system is said to have one of properties defined in the previous section if each element of $\varOmega$ has this property.
Each measure dynamical system  $(\varOmega,T)$ gives rise to the  family  $H_\omega$, $\omega\in \varOmega$, defined by
$$H_\omega = -\Delta + \omega.$$

The  first main  result in this section concerns minimal measure dynamical systems with a finite local complexity property.

\begin{theorem}\label{minimal}  Let $(\varOmega,T)$ be a minimal  measure dynamical system with s.f.d.p. Then, $(\varOmega,T)$ is periodic if one  $H_\omega$ possesses some absolutely continuous spectrum.
\end{theorem}
\begin{proof}  Let $\omega\in\Omega$ be chosen such that $H_\omega$ possesses some absolutely continuous spectrum. Then, the restriction of $H_\omega$ to $[0,\infty)$ or the restriction of $H_\omega$ to $(-\infty,0]$ must possess some absolutely continuous spectrum by general principles. Assume without loss of generality that the restriction of $H_\omega$ to $[0,\infty)$ possesses some absolutely continuous spectrum. Then, $\omega$ is eventually periodic by our main result, Theorem \ref{main}. By minimality, then each element of $\Omega$ must be periodic (with the same period).
\end{proof}

\begin{remark} If the system is minimal then the absolutely continuous spectrum of the $H_\omega$ does not depend on $\omega$ by results of Last /Simon \cite{LaSi} (if the measures possess a density with respect to Lebesgue measure).
\end{remark}

\medskip

Our second main  result in this section deals with ergodic measure dynamical systems. We need a little preparation.

\medskip

\begin{lemma} Let $(\varOmega,S)$ be a dynamical system with ergodic measure $m$. If  the set of  eventually  periodic points has positive measure, then any point in the support of $m$ is periodic (with the same period).
\end{lemma}
\begin{proof} For each compact interval $I$ let $A^I$ be the set of points of $\varOmega$ which are eventually periodic with period $p\in I$. Then, each $A^I$ is closed and invariant under $S$ and hence has measure zero or one by ergodicity. As the  $A^p$  are pairwise disjoint we infer from the assumption  that there exists exactly one $p$ such that $A^p$ has full measure. Then  $S_n A^p$ has full measure for each $n\in \N$ and so has
$$ A^p_\infty := \cap_{n\in \N}  S_n A^p.$$
By construction the elements of $A^p_\infty$ are periodic with period $p$. We have thus found a set of full measure all of whose elements are periodic (with the same period). This implies the statement.
\end{proof}

Whenever  $(\varOmega,T)$ is a measure dynamical system with ergodic measure   $\mu$  then there exists  a closed subset $\Sigma_{ac}$ of $\RR$ such that $\Sigma_{ac}$ is the absolutely continuous spectrum of $H_\omega$ for $\mu$ almost every $\omega \in \varOmega$ (see e.g. the monograph \cite{CL}).

\begin{theorem} \label{ergodic} Let $(\varOmega,T)$ be measure dynamical system of finite local complexity and $m$ an ergodic measure on $\varOmega$.  If $\Sigma_{ac}\neq \emptyset$, then $\omega$ is periodic for every $\omega $ in the support of $m$.
\end{theorem}
\begin{proof}  By our main result Theorem \ref{main} the measure $\omega$ is eventually periodic for almost every $\omega\in \varOmega$. This implies the statement by the previous lemma.
\end{proof}

\section{Systems generated by subshifts}\label{Sec5}
In this section, we discuss a systematic way to generate measure dynamical systems with f.l.p by a suspension type construction over a subshift over a finite alphabet. In the last part of the section we then show how this can be applied to deals with Delone sets.

\bigskip

We start by introducing the necessary notation concerning subshifts. Let $A$, called the alphabet,  be a
finite set equipped with the discrete
topology.  A pair $(X,\tau)$ is a subshift over $A$ if $\Omega$
is a closed subset of $A^{\Z}$, where  $A^{\Z}$ is given the
product topology and $\Omega$ is invariant under the shift operator
$\tau :A^{\Z}\longrightarrow A^{\Z}$, $\tau a(n)\equiv a(n+1)$.
We consider
sequences over $A$ as words and  use
standard concepts from the theory of words, see e.g. \cite{Lot1}.  In particular,
 the number of occurrences
of $v$ in $w$ is denoted by $\sharp_v (w)$ and the length $|w|$ of the word
$w=w(1)\ldots w(n)$ is given by $n$.
Given a subshift $(X,\tau)$ over the finite alphabet $A$ with cardinality $N$, we can associate to each  $\ell =(l_1,\ldots, l_N)\in (0,\infty)^N$ a new  measure dynamical system $(X_\ell, T)$  over $\RR$ by a suspension type construction. More precisely, assume without loss of generality that $A=\{1,\ldots, N\}$ and define  for each $x\in X$   the measure $\xi_x$ by
$$ \xi_x := \sum_{j\in\Z}  x(j) \delta_{\sum_{k=0}^{j-1} l_{x (k)}}.$$
where $\delta_x$ denotes the unit point measure at $s\in\RR$.  Now, set
$$X_\ell :=\{ T_t \xi_x : x\in X, t\in\RR\}.$$
Then, $(X_\ell, T)$ is a measure dynamical system, which inherits many properties of $(X,\tau)$. For example it is not hard to see that minimality of $(X,\tau)$ implies minimality of $(X_\ell,T)$. Similarly, $(X_\ell, T)$ is uniquely ergodic if $(X,\tau)$ is uniquely ergodic. More generally, any invariant measure $n$ on $(X,\tau)$ induces a unique measure $m$ on $(X_\ell,T)$ with
$$n(U_w) |I| = m(U_{w,I})$$
for any finite word $w$ and any sufficiently short interval $I$. Here,
$$U_w :=\{x\in X: w= x(0)\ldots x(|w| -1)$$
and
$$U_{w,I}:=\{\omega \in X_\ell: \mbox{ $w'$ occurs in $\omega$ at an $t\in I$}\},$$
with $w':=\sum_{j=1}^{|w|} w(j) \delta_{\sum_{k=0}^{j-1} l_{x (k)}}.$
A discussion of more or less this phenomenon can be found in \cite{LMS,Klassert}.

After these preparations we can now discuss how to associate measure dynamical systems to subshifts. Let $(X,\tau)$ be a subshift over the finite $A$ with cardinality $N$. Let $\nu:=(\nu_1,\ldots, \nu_N)$ an $N$-tupel of measures with compact support such that
$$\inf \supp \nu_j =0$$
for all $j=1,\ldots, N$. Chose  $l_j:= \sup \supp \nu_j$ if $\nu_j$ is not just a point measure in $0$ and $l_j=1$ otherwise. Now, define for each $x\in X$ the measure $\omega_x$ by
$$\omega_x :=\sum_{j\in\Z} \delta_{\sum_{k=0}^{j-1} l_{x_k}} \ast  \nu_{x(j)}.$$
Thus, we replace in $\xi_x$ any point measure with coefficient $j$ with the measure $\nu_j$. Let
$$\Omega_\nu:=\{ T_t \omega_x : x\in X, t\in\RR\}.$$
By construction $\Omega_\nu$ has $f.d.p$ and by the criteria above, in particular, Proposition \ref{hilf}, we have simple sufficient condition for $s.f.d.p$.
Moreover,  by construction there is a factor map
$$\pi : (X_\ell,T)\longrightarrow (\Omega_\nu,T)$$
such that $\pi (\xi)$ is the measure arising  by replacing any point measure with coefficient $j$ in $\xi$ with the measure $\nu_j$. By standard theory, the factor $(\Omega_\nu,T)$ inherits properties like minimality and unique ergodicity from $(X_\ell,T)$ . Moreover, $\pi$ maps invariant (ergodic) measures on $(X_\ell,T)$  to invariant (ergodic)  measures on $(\Omega_\nu,T)$ (see e.g. \cite{BaakeL-05} for a discussion of these properties).

Putting the above considerations together we have shown the following.

\begin{prop} Let $(X,\tau)$ be a subshift over the finite alphabet $A$. Let $N$ be the cardinality of $A$ and $\nu:=(\nu_1,\ldots, \nu_N)$ be a tuple of  measures with compact support
 and $\inf \nu_j=0$, $j=1,\ldots,N$. Assume that not two of the $\nu_j$ are multiple of Lebesgue measure.
 Then $(\Omega_\nu, T)$ is a measure dynamical system with $s.f.d.p$ and  the following holds:

(a) The  measure dynamical sytem $(\Omega_\nu, T)$ is minimal if $(X,\tau)$ is minimal.

(b)  Any invariant probability measure $n$ on $(X,\tau)$ induces a canonical invariant probability measure $m$ on $(\Omega_\nu, T)$. If $n$ is ergodic, so is $m$.

(c) If $(X,\tau)$ is uniquely ergodic, so is $(\Omega_\nu, T)$.

\end{prop}

Combined with our main result, Theorem \ref{main}, the previous proposition gives the following.

\begin{coro} Let $(X,\tau)$ be a subshift over a finite alphabet with ergodic meausure $n$. Let  $(\Omega_\nu,T)$ with ergodic measure $m$  be associated to $(X,\tau)$ as in the previous proposition. If the almost sure absolutely continuous spectrum of the family $(H_\omega)$, $\omega\in\Omega_\nu$, is not empty, then any $\omega$ from the support of $m$ is periodic.
\end{coro}

\begin{remark}{\rm  The result does not state periodicity of $(X,\tau)$. In fact, such a periodicity can not be concluded. Take for example  all measures $\nu_j$ to be equal $\delta_0$. Then, $(\Omega,T)$ will be periodic rrespective of $(X,\tau)$. }
\end{remark}

\begin{remark} {\rm The discussion in this section was phrased in terms of dynamical systems. However, it is clear that we actually have pointwise constructions viz any sequence over $A$ can be turned into a measure with $s.f.d.p$ by the above procedure.}
\end{remark}

\medskip

Let us here shortly describe two classes of subshifts to which our Corollary can be applied:

\bigskip

\textbf{Example - Bernoulli shifts.} Here, the subshift is given by $X=\{0,1\}^\Z$.
The measure $n$ on $X$  is the product measure $\Pi :=\prod \mu$, where $\mu$ is the measure assigning the value $p$ to $\{0\}$ and the value $1 - p$ to $\{1\}$ for some fixed $p$ with $0< p < 1$. Choosing measures $\nu_1$ and $\nu_2$ with compact support, we can then apply the above theory to conclude absence of absolutely continuous spectrum almost surely whenever the resulting dynamical system is not periodic. Of course, for these more is known and in fact  dynamical localization was shown in  \cite{DSS} (at least when the measures in question possess an $L^1$ density).

\medskip

\textbf{Example - Circle maps.}  Let $\alpha \in (0,1)$ be irrational and $\beta \in (0,1)$ be arbitrary. Let $\Omega (\alpha,\beta)$ be the smallest subshift over $\{0,1\}$ containing the function
$$V_{\alpha,\beta} : \Z\longrightarrow \{0,1\}, \; V_{\alpha,\beta} (n):= 1_{(1-\beta,1]} (n\alpha \mod 1).$$
Then, $(\Omega,\alpha)$ is aperiodic, minimal and uniquely ergodic  as $\alpha$ is irrational (see e.g. Appendix of \cite{DL} for discussion). Again, we can conclude absence of absolutely continuous spectrum almost surely whenever the resulting dynamical system is not periodic.

\bigskip

So far, the considerations in this section have been concerned with subshifts and not with Delone sets.
Let us finish this section, by showing how Delone sets of finite local complexity give rise to sequences over a finite alphabet. Combined with the above discussion this will provide a way to associate measure dynamical systems to Delone dynamical systems. We  only sketch the procedure. Let $\Lambda\subset\RR$ be a Delone set of finite local complexity.  Thus, for each $R>0$ the set
$$A :=\{ (x-R,x+R)\cap \Lambda : x\in\Lambda\}$$
is finite.  Assume without loss of generality that $0\in\Lambda$ and order the points of $\Lambda$ according to $$\ldots x_{-2} < x_{-1} < 0=x_0 < x_1 < x_2\ldots$$
We can then associate to $\Lambda$ the  sequence
$$ x: \Z\longrightarrow A, \; x_j\mapsto \{ (x_j-R,x_j+R)\cap \Lambda.$$
Obviously, $\Lambda$ can be recovered from $x$ and vice versa (see \cite{Len2} for a further discussion).

\section{Absence of point spectrum and purely singular continuous spectrum} \label{Sec6}
In this section, we shortly discuss a criterion for absence of point spectrum. This will allow us to provide some examples with purely singular continuous spectrum. In the discrete case the criterion was introduced by Gordon \cite{Gor}. A version for continuum operators was given by Damanik/Stolz \cite{DamanikStolz}. In the main result of this section our  measures will be assumed to have a density with respect to Lebesgue measure. The condition of translation boundedness then amounts to a uniform $L^1_{loc}$ property. Accordingly, we will speak about $L^1_{loc}$ dynamical systems.

\bigskip

We start by  discussing a condition first used by Kaminaga  \cite{Kam}  in his study of absence of eigenvalues for discrete Schroedinger operators associated to  circle maps. The condition has then be used in various versions in the study of discrete one-dimensional quasicrystals  (see \cite{Damanik} for a survey with further references). We are going to use it in an analogous way.

\begin{definition} A measure dynamical system  $(\Omega,T)$  with ergodic measure $m$ is said to satisfy condition (K) if  there exist  $P_n\geq 0$ with $P_n\to \infty$ such that
$$ G_n:= \{ \omega \in \Omega : \omega_{[0,P_n)} = (T_{P_n} \omega)_{[0,P_n)} = (T_{-P_n} \omega)_{[0,P_n)} \}$$
satisfies
$$\limsup_{n\to \infty} m (G_n) >0.$$
\end{definition}

The relevance of this condition in the study of absence of eigenvalues comes from the so-called Gordon lemma. Here comes a special case of the 'Gordon Lemma' from \cite{DamanikStolz}.

\begin{lemma}  Let $V$ be function such that $V dx $ is a translation bounded measure. If there exist $P_n$-periodic  translation bounded measures of the form $V_n dx$ with $P_n\to \infty$ and
$$ \int_{-P_n}^{2 P_n} | V(x) - V_n (x)| dx =0$$
for all $n$, then $-\Delta + V$ does not have eigenvalues.
\end{lemma}

We are now in a position to state an abstract result on absence of eigenvalues.

\begin{lemma} If the  $L^1_{\mbox{loc}}$  dynamical system  $(\Omega,T)$  with ergodic measure $m$ satisfies condition (K), then $-\Delta  + \omega$ does not have eigenvalues for $m$ almost every $\omega\in\Omega$.
\end{lemma}
\begin{Proof} By ergodicity, the set $\Omega_c$  of $\omega\in \Omega$ with empty point spectrum has either full measure or zero measure. Thus, to prove the statement it suffices to show that it has positive measure. By the previous Gordon lemma, we have that
$$G :=\bigcap_n \bigcup_{k=n}^\infty G_n$$
belongs to $\Omega_c$. By $(K)$ and a short calculation we obtain
$$m(\Omega_c) \geq m(G) \geq \lim_{n\to \infty} m ( \bigcup_{k=n}^\infty G_n ) \geq \limsup m (G_n) >0.$$
This finishes the proof.
\end{Proof}

We can combine the previous lemma with our main result to obtain the following.

\begin{theorem} Let $(\Omega,T)$ be an $L^1_{\mbox{loc}}$  dynamical system with ergodic measure $m$ sastisfying  s.f.d.p and (K). Then, $- \Delta + \omega$ has almost surely purely singular continuous spectrum.
\end{theorem}

For the class of systems generated by subshifts, the condition (K) can be derived whenever the subshift satisfies the analog condition. We say that a subshift $(X,\tau)$ satifies (K), if there exist natural numbers $p_n\to \infty$ such that
$$G_n^X :=\{ x\in X: x(-p_n)\ldots x(-1) = x(0) \ldots x(p_n -1) = x(p_n)\ldots x (2 p_n -1)\}$$
satisfy
$$\limsup_{n\to\infty} (G^X_n) >0$$

\begin{prop} Let $(X,\tau)$ be a subshift over a finite alphabet with ergodic measure $n$, which satisfies (K). Let  $(\Omega_\nu,T)$ be an associated measure dynamical system with measure $m$.
Then $(\Omega_\nu,T)$ satisfies condition (K).
\end{prop}
\begin{Proof} This follows easily from the defining properties of $m$.
\end{Proof}

\textbf{Example - circle maps.} The previous proposition can be applied to all circle maps with $\alpha$ whose continued fraction expansion has infinitely many coefficients of value  at least $4$  \cite{Kam}. The set of such  $\alpha$'s is  a set of full Lebesgue measure.

\bigskip

\medskip

\bigskip

\textbf{Acknowledgements.}  P.S. and D.L. take the opportunity to thank David Damanik for an invitation to a highly enjoyable conference. Useful discussions with Christian Remling on this occasion are also gratefully acknowledged.

\end{document}

%% file: FIG4.pspdftex
\begin{picture}(0,0)%
\includegraphics{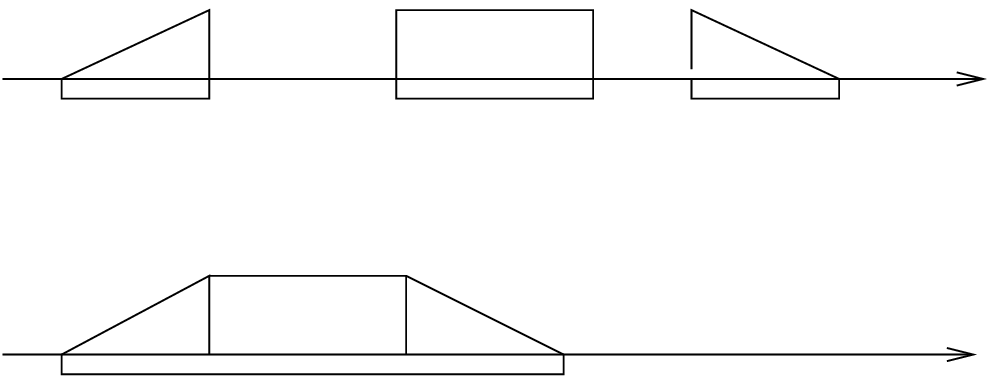}%
\end{picture}%
\setlength{\unitlength}{4144sp}%
\begingroup\makeatletter\ifx\SetFigFont\undefined%
\gdef\SetFigFont#1#2#3#4#5{%
  \reset@font\fontsize{#1}{#2pt}%
  \fontfamily{#3}\fontseries{#4}\fontshape{#5}%
  \selectfont}%
\fi\endgroup%
\begin{picture}(4524,2428)(439,-1970)
\put(856,-646){\makebox(0,0)[lb]{\smash{{\SetFigFont{12}{14.4}{\rmdefault}{\mddefault}{\updefault}{\color[rgb]{0,0,0}$I_1$}%
}}}}
\put(2341,-601){\makebox(0,0)[lb]{\smash{{\SetFigFont{12}{14.4}{\rmdefault}{\mddefault}{\updefault}{\color[rgb]{0,0,0}$I_2$}%
}}}}
\put(3646,-601){\makebox(0,0)[lb]{\smash{{\SetFigFont{12}{14.4}{\rmdefault}{\mddefault}{\updefault}{\color[rgb]{0,0,0}$I_3$}%
}}}}
\put(1126,254){\makebox(0,0)[lb]{\smash{{\SetFigFont{12}{14.4}{\rmdefault}{\mddefault}{\updefault}{\color[rgb]{0,0,0}$\nu_1$}%
}}}}
\put(2341,299){\makebox(0,0)[lb]{\smash{{\SetFigFont{12}{14.4}{\rmdefault}{\mddefault}{\updefault}{\color[rgb]{0,0,0}$\nu_2$}%
}}}}
\put(4141,254){\makebox(0,0)[lb]{\smash{{\SetFigFont{12}{14.4}{\rmdefault}{\mddefault}{\updefault}{\color[rgb]{0,0,0}$\nu_3$}%
}}}}
\put(1801,-1906){\makebox(0,0)[lb]{\smash{{\SetFigFont{12}{14.4}{\rmdefault}{\mddefault}{\updefault}{\color[rgb]{0,0,0}$I$}%
}}}}
\put(3016,-1186){\makebox(0,0)[lb]{\smash{{\SetFigFont{14}{16.8}{\rmdefault}{\mddefault}{\updefault}$\nu=\nu_1|\nu_2|\nu_3$}}}}
\end{picture}%

%% file: FIG5.pspdftex
\begin{picture}(0,0)%
\includegraphics{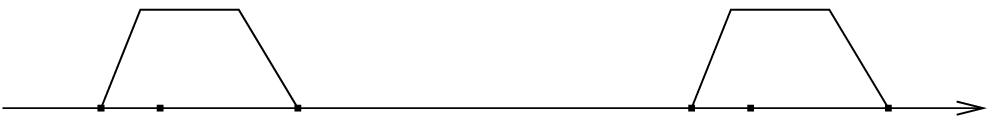}%
\end{picture}%
\setlength{\unitlength}{4144sp}%
\begingroup\makeatletter\ifx\SetFigFont\undefined%
\gdef\SetFigFont#1#2#3#4#5{%
  \reset@font\fontsize{#1}{#2pt}%
  \fontfamily{#3}\fontseries{#4}\fontshape{#5}%
  \selectfont}%
\fi\endgroup%
\begin{picture}(4524,895)(439,-494)
\put(721,-421){\makebox(0,0)[lb]{\smash{{\SetFigFont{12}{14.4}{\rmdefault}{\mddefault}{\updefault}{\color[rgb]{0,0,0}$y_k$    }%
}}}}
\put(1711,-421){\makebox(0,0)[lb]{\smash{{\SetFigFont{12}{14.4}{\rmdefault}{\mddefault}{\updefault}{\color[rgb]{0,0,0}$y_{k+1}=y^1$}%
}}}}
\put(3511,-421){\makebox(0,0)[lb]{\smash{{\SetFigFont{12}{14.4}{\rmdefault}{\mddefault}{\updefault}{\color[rgb]{0,0,0}$y_m$}%
}}}}
\put(4411,-421){\makebox(0,0)[lb]{\smash{{\SetFigFont{12}{14.4}{\rmdefault}{\mddefault}{\updefault}{\color[rgb]{0,0,0}$y_{m+1}=z^1$}%
}}}}
\end{picture}%